%% file: CA_arxiv.tex
\documentclass{article}

\usepackage{algorithm2e}
\usepackage{mdframed}
\usepackage{hyperref}

\usepackage{marvosym}
\usepackage[margin=1in]{geometry}
\newcommand{\envelope}{\raisebox{-.5pt}{\scalebox{1.45}{\Letter}}\kern-1.7pt}

\input{defs.tex}

\newcommand{\K}{\mathcal{K}}
\newcommand{\Ll}{\mathcal{L}}

\newcommand{\I}{\mathcal{I}}

\title{An Algorithmic Approach to the Asynchronous Computability Theorem}
\author{Vikram Saraph \and Maurice Herlihy \and Eli Gafni}

\begin{document}

\maketitle
\begin{abstract}

The asynchronous computability theorem (ACT) uses concepts from combinatorial
topology to characterize which tasks have wait-free solutions in read-write memory.
A task can be expressed as a relation between two chromatic simplicial complexes.
The theorem states that a task has a protocol (algorithm) if and only if there is a certain
chromatic simplicial map compatible with that relation.

While the original proof of the ACT relied on an involved combinatorial argument,
Borowsky and Gafni later proposed an alternative proof that relied on a algorithmic construction,
termed the ``convergence algorithm''.
The description of this algorithm was incomplete, and presented without proof.
In this paper, we give the first complete description, along with a proof of correctness.

\end{abstract}


\section{Introduction}
Herlihy and Shavit's~\cite{HerlihyS93,HerlihyS99} \emph{asynchronous computability
theorem} (ACT) characterizes which tasks are solvable in asynchronous
shared-memory systems.
Informally,
this characterization employs the language of combinatorial topology:
a task can be expressed as a relation between two ``colored'' simplicial
complexes.
The theorem states that a protocol (algorithm) exists if and only if there is a
``color-preserving'' simplicial map from one complex to the other compatible with
the task's relation.
This approach replaces the usual operational model,
where computations unfold in time,
with a static model in which all possible protocol interleavings are
captured in a single combinatorial structure.

The proof of the ACT contains one difficult step.
Using the classical \emph{simplicial approximation
theorem}~\cite[p.89]{Munkres84},
it is straightforward to construct a simplicial map having all desired
properties except that of being color-preserving.
To make this map color-preserving required a rather long construction
employing mechanisms from point-set topology, such as $\epsilon$-balls
and Cauchy sequences.

Borowsky and Gafni~\cite{BorowskyG97} later proposed an alternative
proof strategy for the ACT in which the essential chromatic property was
guaranteed by an algorithm, rather than by a combinatorial construction.
The description of this algorithm was sketchy, however, and no proof was provided.
In this paper, we give the first complete description of their
algorithm,
along with its first proof of correctness.

The paper is structured as follows.
In Section \ref{sec:topology}, we briefly introduce the
combinatorial model and the relevant topological machinery, and we give a short
exposition on the work by Borowsky and Gafni~\cite{BorowskyG97}. In Section 3,
we describe the convergence algorithm and a subroutine used by it. In Section 4,
we give a complete proof of correctness of the convergence algorithm. In Section 5,
we explain how the convergence algorithm can be applied to more general tasks.
In Section 6, we summarize some related work, and in Section 7, we conclude with
some final remarks.

\section{Combinatorial Topology}
\label{sec:topology}
This section reviews the basic notions of combinatorial topology that
will be used to describe shared-memory computation.
Our basic construct, called a simplicial complex,
can be defined in two complementary ways: combinatorial and geometric.
Sometimes it is convenient to use one view, sometimes the other,
and we will go back and forth as needed.

\subsection{Abstract Simplicial Complexes}
Given a finite set $V$,
and a~family $\cK$ of finite subsets of $V$,
we say that $\cK$ is an \emph{abstract simplicial complex} on $V$ if the following hold:
\begin{enumerate}
\item [(1)] if $X\in\cK$, and $Y\subseteq X$, then $Y\in\cK$;
\item [(2)] $\{v\}\in\cK$, for all $v\in V$.
\end{enumerate}
Each set in $\cK$ is called a \emph{simplex},
usually denoted by lower-case Greek letters $\sigma$ and $\tau$.
A subset of a simplex is a \emph{face} of that simplex.
The \emph{dimension} $\dim \sigma$ is
one less than the number of elements of $\sigma$, or $|\sigma| - 1$.
We use ``$n$-simplex'' as shorthand for ``$n$-dimensional simplex'',
and similarly for ``$n$-face''.
A simplex $\phi$ in $\cK$ is a \emph{facet} of $\cK$ if $\phi$ is not
contained in any other simplex.
A complex is \emph{pure} if all its facets have the same dimension.
The \emph{$n$-skeleton} of a complex $\cK$, $\skel^n \cK$,
is the complex formed by all simplexes of $\cK$ of dimension $n$ or less.
If $\cK$ and $\cL$ are complexes with $\cK \subseteq \cL$,
we say $\cK$ is a \emph{subcomplex} of $\cL$.

We use $\Delta^n$ to denote the complex consisting of a single $n$-simplex and
its faces.
A \emph{labeling} is a map $\lambda: V \rightarrow D$,
where $D$ is an arbitrary domain, such as the natural numbers.
A map $\phi$ carrying vertexes of $\cK$ to vertexes of $\cL$ is a
\emph{simplicial map} if it also carries simplexes of $\cK$ to simplexes of
$\cL$.
A simplicial map $\chi: \K \rightarrow \Ll$ is a \emph{coloring} if for each
simplex $\sigma$ of $\cK$, the vertexes of $\sigma$ are mapped to distinct vertexes.
Most complexes considered here are endowed with a coloring $\chi$,
and such complexes are called \emph{chromatic}.
A simplicial map $\phi: \cK \to \cL$ is \emph{color-preserving} if
$\chi(v) = \chi(\phi(v))$ for vertexes $v \in \cK$.

For complexes $\cK$ and $\cL$, a~\emph{carrier map}, written
\begin{equation*}
\Gamma: \cK \to 2^\cL,
\end{equation*}
maps each simplex $\sigma\in\cK$ to a~subcomplex $\Gamma(\sigma)$ of~$\cL$,
such that if $\tau \subseteq \sigma$,
then $\Gamma(\tau) \subseteq \Gamma(\sigma)$.
A simplicial map $\gamma$ is \emph{carried by} a carrier map $\Gamma$ if
$\gamma(\sigma) \subseteq \Gamma(\sigma)$ for every simplex in their domain.
In this case, $\gamma$ is said to be \emph{carrier-preserving}.

The \emph{star} of a simplex $\sigma \in \cC$, written $\Star(\sigma,\cC)$,
or $\Star(\sigma)$ when $\cC$ is clear from context,
is the complex that consist of every simplex $\tau$ which contains $\sigma$,
and every simplex contained in such $\tau$.

\subsection{Geometric Complexes}
In the alternative geometric view,
we embed a complex in a Euclidean space $\bbR^d$ of sufficiently high dimension.

Given a set $X = \set{x_0, \ldots, x_n}$ of points in $\bbR^d$,
their \emph{convex hull}, written $\conv X$,
is the set of points $y$ that can be expressed as
\begin{equation*}
y = \sum_{i=0}^n t_i \cdot x_i,
\end{equation*}
where the coefficients $t_i$ satisfy $0 \leq t_i \leq 1$,
and $\sum_{i=0}^n t_i = 1$.
The $t_i$ are called the \emph{barycentric coordinates} of $y$ with
respect to $X$.
The set $X$ is \emph{affinely independent} if no point in the set can
be expressed as a weighted sum of the others.

A (geometric) vertex is a point in $\bbR^d$,
and a (geometric) $n$-simplex is the convex hull of $(n+1)$
affinely-independent geometric vertexes.
A \emph{geometric simplicial complex} $\cK$ in $\bbR^d$ is a~collection
of of geometric simplexes, such that
\begin{enumerate} 
\item [(1)] any face of a~$\sigma\in\cK$ is also in $\cK$;
\item [(2)] for all $\sigma,\tau\in\cK$, their intersection
 $\sigma\cap\tau$ is a~face of each of them.
\end{enumerate}
We use $|\cK|$ to denote the point-set occupied by the geometric
complex $\cK$.

Given a geometric simplicial complex $\cK$, we can define the
underlying abstract simplicial complex $\cC(\cK)$ as follows: take the
union of all the sets of vertexes of the simplexes of $\cK$ as the
vertexes of $\cC(\cK)$, then for each simplex
$\sigma=\conv\{v_0,\dots,v_n\}$ of $\cK$ take the set
$\{v_0,\dots,v_n\}$ to be a simplex of $\cC(\cK)$.

In the opposite direction,
given an abstract simplicial complex $\cA$ with finitely
many vertexes, there exist many geometric simplicial complexes $\cK$,
such that $\cC(\cK)=\cA$.

The \emph{open star}, denoted $\Ostar(\sigma)$,
is the union of the interiors of the simplexes that contain $\sigma$:
\begin{equation*}
\Ostar(\sigma)=\bigcup_{\tau\supseteq\sigma}\inte\tau.
\end{equation*}
Note that $\Ostar(\sigma)$ is not an~abstract or geometric simplicial
complex,
but just a open point-set in $\cC$.

\subsection{Subdivisions}
\begin{figure}
\centerline{\includegraphics[width=\hsize]{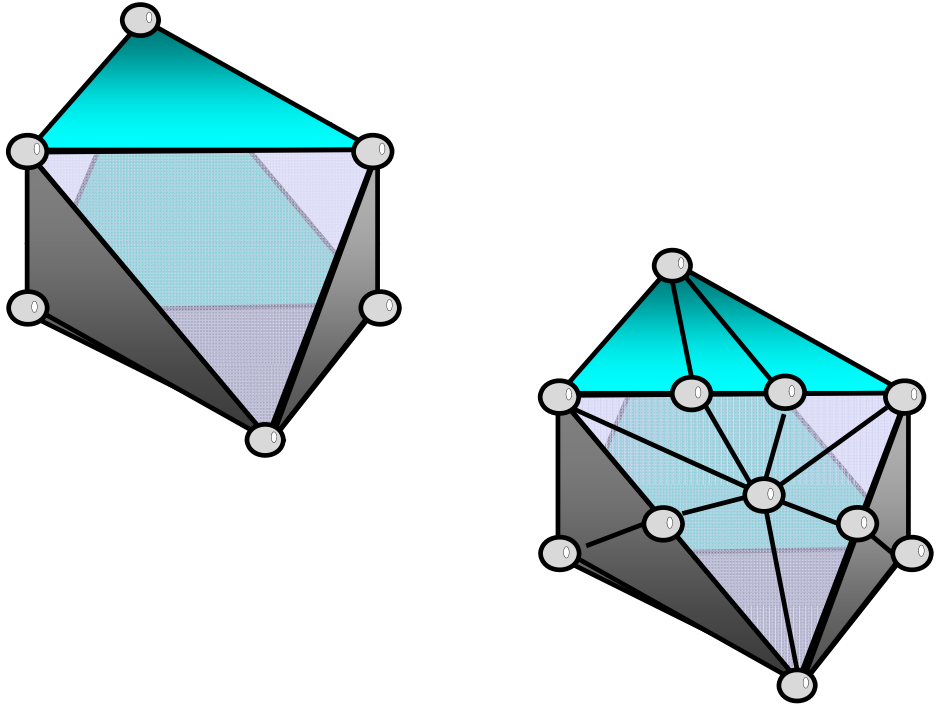}}
\caption{A geometric complex and a subdivision of it.}
\label{figure:ch03:subdivision}
\end{figure}

If $\cA$ and $\cB$ are geometric complexes,
a continuous map $f: |\cA| \to |\cB|$ is \emph{carried by}
a carrier map $\Phi: \cA \to 2^{\cB}$ if, for every simplex $\sigma \in \cA$,
$f(\sigma) \subseteq |\Phi(\sigma)|$.

Informally, a \emph{subdivision} of a complex $\cK$ is
constructed by ``dividing'' the simplexes of $\cK$ into smaller
simplexes, to obtain another complex $\cL$. 
Subdivisions can be defined for both geometric and abstract complexes.

\begin{definition}
\label{def:ch03:subdivision}
\index{subdivision}
A geometric complex $\cB$ is called a \emph{subdivision} of a geometric
complex $\cA$ if the following two conditions are satisfied:
\begin{enumerate}
\item[(1)] $|\cA|=|\cB|$; 
\item[(2)] each simplex of $\cA$ is the union of finitely many simplexes of~$\cB$.
\end{enumerate}
\end{definition}
Figure~\ref{figure:ch03:subdivision} shows a geometric complex and
a~subdivision of that complex.

Subdivisions can be defined for abstract simplicial complexes as well. 
\begin{definition}\label{def:ch03:Subd}
Let $\cA$ and $\cB$ be abstract simplicial complexes. We say that
$\cB$ \emph{subdivides} the complex $\cA$ if there exists
a~homeomorphism $h:|\cA|\rightarrow|\cB|$ and a~carrier map
$\Phi:\cA\rightarrow 2^{\cB}$, such that for every simplex
$\sigma\in\cA$, the restriction $h|_{|\sigma|}$ is a~homeomorphism
between $|\sigma|$ and $|\Phi(\sigma)|$.
\end{definition}

A subdivision $\Div(\cK)$ is \emph{chromatic},
if $\Div(\cK)$ is chromatic,
and for every simplex $\sigma \in \cK$,
$\Div(\sigma)$ and $\sigma$ have the same colors.

Let $\Div \cK$ be a subdivision.
The \emph{carrier} of a simplex $\sigma$, $\Car(\sigma)$,
is the smallest simplex $\kappa \in \Div \cK$ such that $\sigma \in \Div(\cK)$.

\subsubsection{Barycentric Subdivision}
\label{section:ch03:barycentric}
In classical combinatorial topology,
the following barycentric subdivision is perhaps the most widely used.
\begin{definition}\label{def:ch03:barySubd}
\index{barycentric subdivision}
\index{subdivision!barycentric}
Let $\cK$ be an~abstract simplicial complex.
Its \emph{barycentric} subdivision $\bary\cK$ is the abstract simplicial
complex whose vertexes are the non-empty simplexes of $\cK$.
A~$(k+1)$-tuple $(\sigma_0,\dots,\sigma_k)$ is a~simplex of $\bary\cK$ if and
only if  the tuple can be indexed so that
$\sigma_0 \subset \cdots \subset \sigma_k$.
\end{definition}

\subsubsection{Standard Chromatic Subdivision}
\begin{figure}
\centerline{\includegraphics[width=\hsize]{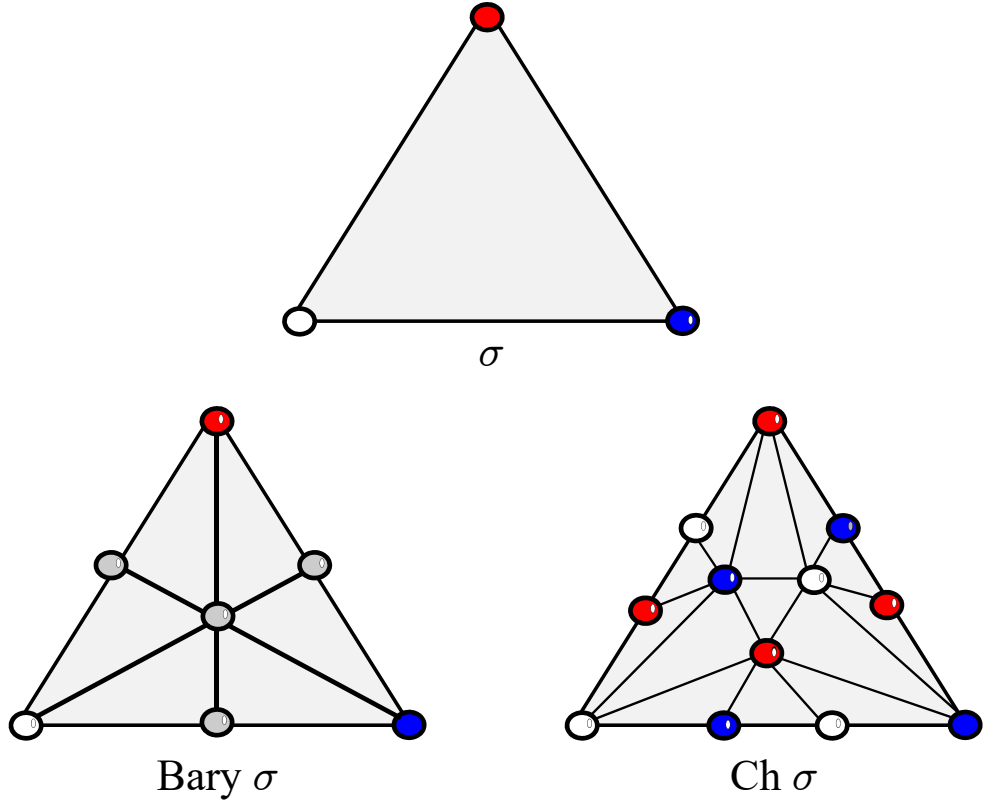}}
\caption{A simplex $\sigma$ (top),
the barycentric subdivision $\bary \sigma$ (bottom left),
and the standard chromatic subdivision $\Ch \sigma$ (bottom right).}
\label{figure:ch03:subdivisions}
\end{figure}

The \emph{standard chromatic subdivision} $\Ch(\cK)$
is the chromatic analog to the barycentric subdivision.
(See Figure~\ref{figure:ch03:subdivisions}.)

\begin{definition}\label{def:ch03:chromSubd}
Let $(\cA,\chi)$ be a~chromatic abstract simplicial complex.
Its \emph{standard chromatic} subdivision $\Ch(\cA)$ is the abstract simplicial
complex whose vertexes have the form $(i,\sigma_i)$,
where $i \in \{0, ...,n\}$, $\sigma_i$ is a non-empty face of $\sigma$,
and $i \in \chi(\sigma_i)$.
A~$(k+1)$-tuple $(\sigma_0,\dots,\sigma_k)$ is a~simplex of $\Ch(\cA)$ if and
only if
\begin{itemize}
\item
the tuple can be indexed so that
$\sigma_0\subseteq \cdots \subseteq \sigma_k$.
\item
for $0 \leq i,j \leq n$,
if $i \in \chi(\sigma_j)$ then $\sigma_i \subseteq \sigma_j$. 
\end{itemize}
Finally, to make the subdivision chromatic,
we define the coloring $\chi: \Ch(\cK)$ to be $\chi(i,\sigma) = i$.
\end{definition}
This subdivision extends to complexes in the obvious way,
and it can be \emph{iterated} to produce subdivisions $\Ch^N(\sigma)$,
for $N > 0$.


\subsection{Links and Connectivity}

Subdivided simplexes exhibit an important property called \emph{link-connectivity}.
We provide the necessary definitions below.

\begin{definition}
The \emph{link} of a $\sigma$ in complex $\cK$,
denoted $\Lk(\sigma, \cK)$,
is the subcomplex of $\cK$ consisting of all simplexes $\tau$ disjoint
from $\sigma$ such that $\tau \cup \sigma \in \cK$.  
\end{definition}


One can think of the link of $\sigma$ as a simplicial neighborhood
that encompasses $\sigma$. Next, we present the concept of \emph{connectivity}
as taken from algebraic topology.

\begin{definition}
Let $S^k$ denote the $k$-dimensional sphere.
A~complex $\cK$ is \emph{$k$-connected} if,
for all $0\leq m\leq k$, any continuous map
$f:~S^m~\to~|\cK|$ can be extended to a continuous $F:~D^{m+1}~\to~|\cK|$, where
the sphere $S^m$ is the boundary of the disk $D^{m+1}$.
\end{definition}

One way to think about this property is that any map $f$ of the $k$-sphere
that cannot be ``filled in'' represents a $k$-dimensional ``hole'' in the complex.  

\begin{definition}
A pure $n$-dimensional complex $\cK$ is \emph{link-connected} if for all $\sigma \in \cK$, $\Lk(\sigma, \cK)$
is $(n - \dim(\sigma) - 2)$-connected.
\end{definition}
Informally, link-connectivity ensures that a complex cannot be ``pinched'' too
thinly. All subdivided simplexes are link-connected, which is an important property for proving the main theorem.

\subsection{Useful Theorems}
We will need the following \emph{pasting lemma} from point-set topology.
\begin{lemma}
\label{lemma:pasting}
Let $X_1$ and $X_2$ be closed sets of a common space, with continuous $f_i : X_i \rightarrow Y$ that coincide on the intersection $X_1 \cap X_2$. Then the
function $f : X_1 \cup X_2 \rightarrow Y$, defined in terms of the $f_i$, is a continuous function.
\end{lemma}

We will also need the following version of the classical
\emph{simplicial approximation theorem}.
\begin{definition}
\label{def:ch03:simplicial-approximation}
Let $\cA$ and $\cB$ be abstract simplicial complexes, let
$f:~|\cA|~\to~|\cB|$ be a~continuous map, and let
$\varphi:~\cA~\to~\cB$ be a~simplicial map. The map $\varphi$ is
called a~\emph{simplicial approximation} to~$f$, if for every simplex
$\alpha$ in $\cA$ we have
\begin{equation}
\label{eq:defsimpappr}
f(\inte|\alpha|) \subseteq \bigcap_{a \in \alpha} \Ostar(\varphi(a))
=\Ostar(\varphi(\alpha)),
\end{equation}
where $\inte|\alpha|$ denotes the interior of $|\alpha|$.
\end{definition}

\begin{theorem}
\label{theorem:ch03:approx}
Let $\cA$ and $\cB$ be simplicial complexes.
Given a~continuous map $f:|\cA|~\to~|\cB|$, there is an
$N > 0$ such that $f$ has a~simplicial approximation
$\varphi:~\Ch^N(\cA)~\to~\cB$.
\end{theorem}
This theorem remains true if we replace the standard chromatic
subdivision $\Ch(\cdot)$ with the barycentric subdivision $\Bary(\cdot)$.

\section{Read-Write Memory}
\input{computation.tex}

\section{The Asynchronous Computability Theorem}
\input{theorem.tex}

\section{Correctness of the Convergence Algorithm}
\input{correctness.tex}

\section{Related work}
A more complete treatment of topological models for concurrent computing can be
found in the textbook of Herlihy, Kozlov, and Rajsbaum~\cite{HerlihyKR2013}.

The original ACT~\cite{HerlihyS93,HerlihyS99} used combinatorial arguments to
construct the color-preserving map required by the theorem.
Borowsky and Gafni~\cite{BorowskyG97} proposed the alternative algorithmic
approach to constructing this map,
but without complete definitions or a proof of correctness.
Guerraoui and Kuznetsov~\cite{GuerraouiK2004} compare the two approaches.
Gafni \emph{et al.}~\cite{GafniKM2014} recently generalized the ACT to
encompass infinite executions and other models.

The first application of the ACT was to prove the impossibility of the
\emph{$k$-set agreement} task~\cite{Chaudhuri90},
a result also proved, using other techniques,
by Borowsky and Gafni~\cite{BorowskyG93},
and by Saks and Zaharoglou~\cite{SaksZ93}.
These results generalize the classic proofs of the impossibility of consensus
due to Fischer \emph{et al.}~\cite{FischerL82} and Biran \emph{et
al.}~\cite{BiranMZ90}.

Castaneda and Rajsbaum~\cite{CastanedaR08} use the ACT to show that the
\emph{renaming} task~\cite{AttiyaBDKPR87} for $n+1$ processes
has no wait-free read/write protocol with $2n$ output names when $n+1$ is a
prime power, but that a protocol does exist when $n+1$ is not a prime power.
Attiya \emph{et al.}~\cite{AttiyaCHP2013} and Kozlov~\cite{Kozlov15a} give
upper bounds on the running times of such protocols.

\section{Remarks}
The heart of the ACT is the construction of a chromatic map from a
chromatic subdivision of the input complex to the output complex.
While it is easy to construct a non-chromatic map using
the well-known \emph{simplicial approximation theorem}~\cite[p.89]{Munkres84},
the only prior construction~\cite{HerlihyKR2013,HerlihyS99} was long and complex.
The algorithmic approach proposed by Borowsky and Gafni~\cite{BorowskyG93} was
intuitively appealing, but lacked a complete statement of the algorithm and a
proof.
Here, we have given such a statement and proof,
and we believe the result, while far from simple,
is easier to follow, and yields new insight into this key construction.

\vspace{10pt}

\newpage
\bibliographystyle{plain}
\bibliography{CA_submission}


\end{document}

%% file: defs.tex
\usepackage{amscd}
\usepackage{amsfonts}
\usepackage{amsmath}
\usepackage{amssymb}
\usepackage{amsthm}
\usepackage{calc}
\usepackage{xcolor}
\usepackage{epsfig}
\usepackage{graphicx}
\usepackage{hyperref}
\usepackage{import}
\usepackage{latexsym}
\usepackage{listings}
\usepackage{url}
\usepackage{verbatim}
\usepackage{framed}
\numberwithin{equation}{section}

\newcommand{\remove}[1]{}

\theoremstyle{remark}
\newtheorem{theorem}{Theorem}[section]

\newtheorem{definition}[theorem]{Definition}

\newtheorem{lemma}[theorem]{Lemma}

\def\cA{\ensuremath{{\cal A}}}   \def\cB{\ensuremath{{\cal B}}}   \def\cC{\ensuremath{{\cal C}}}   
         
\def\cI{\ensuremath{{\cal I}}}      \def\cK{\ensuremath{{\cal K}}}   \def\cL{\ensuremath{{\cal L}}}
      \def\cO{\ensuremath{{\cal O}}}   \def\cP{\ensuremath{{\cal P}}}

\DeclareMathOperator{\Bary}{Bary}

\DeclareMathOperator{\Car}{Car}

\DeclareMathOperator{\Ch}{Ch}

\DeclareMathOperator{\Div}{Div}

\DeclareMathOperator{\Lk}{Lk}
\DeclareMathOperator{\Ostar}{St^\circ}

\DeclareMathOperator{\Star}{St}

\DeclareMathOperator{\bary}{Bary}

\DeclareMathOperator{\conv}{conv}

\DeclareMathOperator{\inte}{Int}
\DeclareMathOperator{\im}{Im}

\DeclareMathOperator{\skel}{skel}

\DeclareMathOperator{\views}{views}

\newcommand{\bbR}{\ensuremath{\mathbb{R}}}

\newcommand{\var}[1]{\lstinline+#1+}

\newcommand{\set}[1]{\left\{ #1 \right\}}

%% file: computation.tex
\label{sec:computation}
In this section we describe and motivate our model of computation.
A \emph{distributed system} is a collection of sequential computing entities,
called \emph{processes}, that cooperate to solve a problem, called a \emph{task}.
The processes communicate by reading and writing a shared memory.
Each process runs a program that defines how and when it communicates
with other processes.
Collectively these programs define a concurrent algorithm, or \emph{protocol}.

The theory of distributed computing is largely about what is computable
in the presence of timing uncertainty and failures.
Here, we assume processes are subject to \emph{crash failures},
in which a faulty process simply halts and falls silent.
We focus on \emph{wait-free} protocols that solve particular
tasks when any proper subset of the of processes may fail.
We adopt an \emph{asynchronous} timing model,
where processes run at arbitrary, unpredictable speeds, and
there is no bound on process step time.
Note that a failed process cannot be distinguished from a slow process.

\subsection{Processes}
As stated, a process is an automaton.
Each process has a unique \emph{name}.
In a computation,
it starts with an input value,
takes a finite number of steps,
and halts with an output value.
The task specification defines which values can be assigned as inputs,
and which can be accepted as outputs.

A process state is modeled as a vertex,
labeled with both that process's name and its state.
A system state is a simplex,
where each vertex represents a process and its state,
and the vertexes represent mutually compatible states.
Such is simplex colored by process names, and labeled with process states.
An initial system state is represented as a simplex whose set of
vertexes represent the simultaneous states of distinct processes.
The set of all possible states forms a chromatic complex,
colored by process names, and labeled with process states.
Henceforth, we speak of a vertex's \emph{color} as shorthand for the
name of the corresponding process.

\subsection{Tasks}
To reduce computation to its simplest form,
we consider a fundamental unit of computation called a~\emph{task}.
An input to a task is distributed:
only part of the input is given to each process.
The output from a task is also distributed:
only part of the output is computed by each process.
The task specification states which outputs can be produced in response to each
inputs.

A \emph{task} is a triple $(\cI, \cO, \Gamma)$,
where $\cI$ is the \emph{input complex\/} that defines all possible initial configurations,
$\cO$ is the \emph{output complex} that defines all possible final configurations,
and $\Gamma$ is a carrier map carrying each $m$-simplex of $\cI$ 
to a subcomplex of $m$-simplexes of $\cO$, for $0 \leq m \leq n$.
$\Gamma$ has the following interpretation:
for each $\sigma \in \cI$,
if the $(\dim \sigma +1)$ processes in $\sigma$ start with the designated
input values, 
and the remaining $n- \dim \sigma$ processes fail without taking any steps,
then each simplex in $\Gamma(\sigma)$ corresponds to a legal final state
of the non-faulty processes.

A task has the following interpretation.
For each simplex $\sigma \in \cI$,
if the $(\dim \sigma +1)$ processes in $\sigma$ each starts on the
unique vertex colored with its name,
taking as input that vertex's labeled value,
and the remaining $n- \dim \sigma$ processes fail without taking any steps,
then each simplex in $\Gamma(\sigma)$ corresponds to a legal final state
of the non-faulty processes,
where each process halts on the unique vertex colored with its name,
taking as output that vertex's labeled value.

The most common task considered in this paper is a \emph{convergence} task.
We are given a input complex $\cI$ colored by process names,
and a chromatic subdivision of $\cI$.
In the task $(\cI, \Div(\cI), \Div)$,
the $n+1$ processes start on vertexes of an $n$-simplex $\sigma$ of $\cI$,
where each process's input vertex is colored by that process's name.
The processes halt on a single simplex of $\Div(\sigma)$,
and each process halts on an vertex is colored by that process's name.

In some circumstances,
we relax the requirement that each processes halt on an output vertex
labeled with its own name.
A \emph{colorless} task is also defined by a triple $(\cI,\cO,\Gamma)$,
but there is no requirement that the decision map be color-preserving.

\subsection{Protocols}
A \emph{protocol} is a concurrent algorithm to solve a task:
initially each process knows its own part of the input, but not the others'.
Each process communicates with the others by reading and writing a
shared memory,
and eventually halts with its own output value.
Collectively, the individual output values form the task's output.

Here, we are concerned with protocols where processes communicate by
reading and writing shared memory.
Without loss of generality,
we consider \emph{full-information} protocols in which each process repeatedly
writes its current state to the shared memory,
and constructs its new state by taking a snapshot of (instantaneously
reading) the states written by other processes.
After a finite number of such communication rounds,
each process applies a \emph{decision map} to its state to choose an output value.

All possible executions of a protocol can be modeled as a
chromatic simplicial complex.
We think of the protocol executions as a chromatic carrier map $\Xi$ 
from the input complex $\cI$ to $\cP$, the chromatic \emph{protocol complex.}
The carrier map $\Xi$ carries each input simplex $\sigma\in\cI$
to a subcomplex $\Xi(\sigma)$ of the protocol complex.
$\Xi$ carries each $\sigma\in\cI$ to the executions in which the
processes in $\sigma$ do not hear from any other process.

The protocol complex is related to the output complex by the decision map $\delta$, 
a chromatic simplicial map that sends each vertex $v$ in the protocol complex 
to a vertex $w$ in the output complex,
colored with the same process name.
The decision value labeling $w$ is the value which the corresponding
process takes as its output value. 
A protocol with carrier map $\Xi$
solves the task if the carrier map obtained by composing $\delta$
with $\Xi$ is carried by $\Delta$. 

\subsection{Read-Write Memory}

%
%
%

In the most natural shared-memory model,
the processes share an ordered sequence of \emph{registers},
memory units that can hold values from an arbitrary domain.
A process can \emph{read} a register,
returning its current contents,
or it can \emph{write} a new value to that register,
obliterating that register's prior value.

There are many variations of the asynchronous read-write model,
all computationally equivalent in the sense that one model can
simulate the others with overhead polynomial in the number of
processes.
Here, we will use two distinct alternatives.
In the \emph{snapshot} model,
each process can write atomically to a single register,
but it can atomically read any set of registers,
an operation called a \emph{snapshot}.

It is also sometimes convenient to use a variation known as the
\emph{immediate snapshot} model~\cite{BorowskyG93,SaksZ93}
An \emph{immediate snapshot} takes place in two contiguous steps.
In the first step, a process writes its view to a word in memory,
possibly concurrently with other processes.
In the very next step, it takes a snapshot of some or all of the memory,
possibly concurrently with other processes.
It is important to understand that that in an immediate snapshot,
the snapshot step takes place \emph{immediately after} the write step.

Although, modern multicores do not support either the snapshot or the
immediate snapshot model directly,
either model can be simulated by more conventional models,
and vice-versa \cite[Ch. 14]{HerlihyKR2013}.
The snapshot model is convenient for expressing specific algorithms,
while the immediate snapshot model is convenient because
the protocol complex generated by an $N$-round immediate snapshot
execution is the subdivision $\Ch^N(\cI)$ of the input complex
\cite[Ch. 4]{HerlihyKR2013}.
(Using individual reads and writes,
the protocol complexes are not subdivisions,
although the can be retracted to subdivisions.)

%% file: theorem.tex
\label{sec:act}
The ACT states that a task $T = (\cI, \cO, \Gamma)$ has a wait-free protocol in
read-write memory exactly when there is a chromatic, simplicial map from
a chromatic subdivision of $\cI$ to $\cO$ carried by $\Gamma$.
\begin{theorem}
\label{theorem:ACT}
A task $(\cI, \cO, \Gamma)$ has a wait-free read-write protocol if and
only if there is a chromatic subdivision $\Div(\I)$ and a
color-preserving simplicial map 
\begin{equation*}
\mu: \Div(\I) \to \cO
\end{equation*}
carried by $\Gamma$.
\end{theorem}

In one direction, the claim is relatively easy.
If there exists a read-write protocol,
then there exists an immediate snapshot protocol
whose protocol complex is $\Ch^N(\cI)$ for some $N > 0$.
The color-preserving simplicial decision map $\delta:~\Ch^N(\cI)~\to~\cO$ is the desired map.

The other direction is more difficult:
we must show that given a chromatic map
$\mu:~\Div \cI~\to~\cO$ carried by $\Div$, for some $N > 0$.
A protocol can be constructed by solving the convergence task on $\Ch^n(\cI)$,
and then using $\phi$ as the decision map.

The most straightforward strategy is to show there exists a color-preserving simplicial map 
\begin{equation*}
\phi: \Ch^N \cI \to \Div \cI,
\end{equation*}
for some $N > 0$, such that for all
$\sigma\in\cI$, $\phi(\Ch^N\sigma)\subseteq\Div\sigma$.
These maps compose as follows:
\begin{equation*}
\Ch^N \cI \stackrel{\phi}{\to} \Div \cI \stackrel{\mu}{\to} \cO.
\end{equation*}
These maps can used to construct a protocol.
From an input simplex $\sigma$, each process performs the following three steps:
\begin{enumerate}
\item[step 1.]
execute an $N$-layer immediate snapshot protocol,
halt on a~vertex $x$ of the simplicial complex $\Ch^N \sigma$, 

\item[step 2.]
compute $y = \phi(x)$, yielding a vertex in $\Div \sigma$, 

\item[step 3.]
compute $z = \mu(y)$, yielding an output vertex.
\end{enumerate}
It is easy to check that all processes halt on the vertexes of a single simplex
in $\Gamma(\sigma)$.
Moreover, because all maps are color-preserving,
each process halts on an output vertex of matching color.

In prior work,
the map $\phi:~\Ch^N \cI~\to~\Div \cI$ was constructed in the following combinatorial way.
The continuous identity map $|\Ch^N \cI| \to |\Div \cI|$
has a simplicial approximation $\psi:~\Ch^N~\cI~\to~\Div~\cI$,
carried by $\Delta$.
The bulk of the proof is concerned with ``perturbing'' $\psi$ to make it color-preserving,
a somewhat long and delicate construction.

Borowsky and Gafni suggested an alternative approach:
treat this problem as a task, $(\cI, \Div \cI, \Div)$,
in which processes start on vertexes of matching color on a simplex $\sigma$ of $\cI$,
and halt on vertexes of matching color on a single simplex of $\Div(\sigma)$.
An explicit protocol that solves this \emph{chromatic simplex agreement} (CSA) task
induces the desired map.
As noted, the original paper lacked a complete description of the algorithm and a proof,
both of which are presented here.

\subsection{Chromatic Simplex Agreement}

Borowsky and Gafni first consider a more simple task, called non-chromatic simplex agreement (NCSA), in which processes start on the same input complex $\I$ and must converge to any simplex of $\Div(\I)$. Processes do not have to land on vertexes of any specified color, though they must remain where they are if they run solo.

NCSA can be solved provided that $\I$ is sufficiently connected. Borowsky and Gafni present a sketch of an inductively constructed protocol, which can be made rigorous via combinatorial topology. While we do not give an explicit protocol here, in the next section we give an explicit protocol for a similar task called \emph{link-based NCSA}, or LNCSA, which is directly used in the construction of the convergence algorithm.

Algorithm~\ref{alg:csa} shows the Chromatic Simplex Agreement (CSA) protocol, which is round-based.
Recall that participating processes begin on a simplex of $\cI$.
In the first round, the processes to run a
simplex agreement protocol on $\Bary(\Div \cI)$,
so that they collectively choose a nested sequence of simplexes in
$\Div \cI$. Each process writes its simplex to shared memory and takes an immediate snapshot
of the simplexes written by the others.
If a process sees a vertex of its color in all simplexes in its snapshot,
then the process returns that vertex and finishes.
Otherwise, it computes its \emph{view}\footnote{
Borowsky and Gafni called the view the ``core''.}
as the union of all simplexes it saw, minus the vertex of its own color, if it exists. The process proceeds to the next round.

\vspace{10pt}

\begin{algorithm}[H]
\begin{mdframed}
\caption{The convergence algorithm}

\SetAlgoLined
\SetKw{shared}{shared}
\SetKwArray{participating}{participating}
\SetKwArray{views}{views}
\SetKwArray{simplexes}{simplexes}
\SetKwProg{protocol}{protocol}{:}{end}
\SetKwFunction{snapshot}{snapshot}
\SetKwFunction{sagree}{simplexAgree}
\SetKwFunction{lncsa}{linkNonchromaticSimplexAgree}
\SetKwBlock{immediate}{immediate}{}
\SetKwFor{uWhile}{while}{do}{}

\shared{\participating{n+1}}\;
\shared{\views{n+1}}\;
\shared{\simplexes{n+1, n+1}}\;
\protocol{chromaticSimplexAgree(p, v, $\cI$, $\Div$)}{
    \participating{p} := $p$\;
    $s$ := \sagree{v, $\cI$, $\Div$}\;
    \immediate{
        \simplexes{1, p} := ($s$, $\varnothing$)\;
        $snap$ := \snapshot{\simplexes{1}}\;
        \uIf{$\exists u : u \in \bigcap_{t \in snap} t$ \text{and} $\chi(u) = p$}{
            \Return{u}\;
        }
        \uElse{
            $toss$ := $\{u : u \in \bigcup_{t \in snap} \text{ and } \chi(u) = p\}$\;
            $w$ := $\bigcup_{t \in snap} t - toss$\;
        }
    }
    $r$ := 2\;
    \uWhile{True}{
        \immediate{
            \views{r, p} := $w$\;
            $snap$ := \snapshot{\views{r}}\;
            $P$ := \snapshot{\participating}\;
        }
        $c$ := $\bigcap_{x \in snap} x$\;
        $s, \overline{c}$ := \lncsa{$v$, $c$, $P$, $\cI$, $\Div$}\;
        \immediate{
            \simplexes{r, p} := ($s$, $\overline{c}$)\;
            $snap$ := \snapshot{\simplexes{r}}\;
        }
        \uIf{$\exists u : u \in \bigcap_{t \in snap} t \text{ and } \chi(u) = p$}{
            \Return{u}\;
        }
        \uElse{
            $toss$ := $\{u : u \in \bigcup_{t \in snap} \text{ and } \chi(u) = p\}$\;
            $w$ := $\bigcup_{t \in snap[0]} t \cup \bigcap_{d \in snap[1]} d - toss$\;
            $r$ := $r + 1$\;
        }
    }
}
\end{mdframed}
\label{alg:csa}
\end{algorithm}

\vspace{10pt}

A process entering round $r>1$ has already computed a view.
Each process begins round $r$ by writing its view to shared memory and
taking immediate snapshots of the shared memory from this round as
well as the first round.
By retaking a snapshot of the first round,
a process updates the \emph{participating set} of processes it
observes running the protocol.
By taking a snapshot of the views from the current round,
the process computes its \emph{core}\footnote{
Borowsky and Gafni called this the ``intersection of cores''.},
defined to be the intersection of the views from the current round's snapshot.
Using these two snapshots,
the process computes its \emph{convergence complex},
a simplicial neighborhood of the core (defined in the next section).
Note that processes typically construct different convergence complexes,
though all such complexes will be ordered by inclusion.
Each process then chooses a \emph{starting vertex} of its own color in
its convergence complex,
and runs the LNCSA protocol
on the barycentric subdivision of its convergence complex.
Processes collectively obtain a nested sequence simplexes, similar to the first round.
Each process then replaces its core with the smallest core observed when running LNCSA.
Each process writes its simplex and its new core to shared memory and takes a snapshot.
If a process sees a vertex of its color in each simplex it saw, it decides on that vertex and finishes.
Otherwise, it updates its view and proceeds to the next round.

See Algorithm \ref{alg:csa} for pseudocode. In the next section, we give more detailed explanations about how each piece of data is computed.

%% file: correctness.tex
\label{sec:correctness}
\subsection{Preliminaries}
If each process executes Protocol~\ref{alg:csa},
it will decide on a vertex of its color in a simplex on $\Div(\cI)$,
thus solving CSA over $\cI$.
In this section we show that the protocol terminates, 
and that processes choose valid outputs.

Here is some notation useful to prove correctness of the convergence algorithm.
Let $p_1, \ldots, p_{n+1}$ be the processes participating in the
convergence algorithm.
Let $p$ be any  such process.
For each round $r > 1$, $p$'s state consists of the following:
its \emph{participating set} $P^r_p$,
its \emph{core} $c^r_p$,
its \emph{convergence complex} $\cC^r_p$,
its \emph{starting vertex} $v^r_p$,
its \emph{simplex} $s^r_p$,
and its view $w^r_p$.

Here is how each component is computed during round $r$ of the convergence algorithm: 
\begin{enumerate}
\item The \emph{participating set} $P^r_p$ is the set of processes and
corresponding input vertexes that process $p$ sees when it takes a
snapshot of the first round's shared memory. 

\item The \emph{core} $c^r_p$ is the set of vertexes that may
be decision values of processes other than $p$ that have finished
the protocol.
It is defined to be
\begin{equation*}
c^r_p = \bigcap_{j \in J} w^r_{p_j},
\end{equation*}
where $J$ is the index set of processes seen by $p$
in its snapshot of views from the current round $r$.
When running LNCSA, 
the process may see smaller cores,
in which case $p$ recomputes its core as
\begin{equation*}
\bar{c}^r_p = \bigcap_{k \in K} c^r_{p_k}
\end{equation*}
where $I$ is the index set of processes seen by $p$ during LNCSA. 

\item The \emph{convergence complex} $\cC^r_p$ is the complex on which
$p$ runs the LNCSA protocol to choose a new simplex consistent with
current decision values.
It is computed as 
\begin{equation*}
\cC^r_p = \Lk(\bigcap_{k \in K} c^r_{p_k}, \bigcup_{k \in K} P^r_{p_k}).
\end{equation*}

\item The \emph{starting vertex} $v^r_p$ is the vertex on which $p$
begins running the LNCSA protocol.
Any vertex from $\cC^r_p$ with the same color as $p$ may be chosen. 

\item The \emph{simplex} $s^r_p$ is the simplex in $\cC^r_p$ which $p$
chooses as a result of running the LNCSA protocol.
They collectively represent the domain of values from which processes may
choose decision values in round $r$. 

\item The \emph{view} $w^r_p$ is the set of vertexes that $p$
sees in round $r$. 
It is computed as
\begin{equation*}
w^r_p = \bigcup_{i \in I} s^r_{p_i} \cup \bigcap_{i \in I} \bar{c}^r_{p_i} - \{u^r_p\}
\end{equation*}
where $I$ is the index set of processes seen by $p$ in its snapshot of
simplexes and cores,
and $u^r_p$ is the vertex with the same color as $p$, if it exists.  
\end{enumerate}

\subsection{Link-based non-chromatic simplex agreement}
Before we can define LNCSA, we need a few technical lemmas.

\begin{lemma}
\label{lemma:simplexes}
All views and cores are simplexes.
\end{lemma}

\begin{proof}
By induction.
It is clear that views from round $1$ are simplexes,
since they are all subsets of the largest simplex chosen during
simplex agreement.
The cores computed in round $2$ are also simplexes,
since they are intersections of the views from round $1$.

Inductively assume that all views from round $r$ are simplexes.
Clearly all cores $c^{r+1}_p$ in round $r+1$ are simplexes as well,
since they are intersections of views from round $r$.
Views in round $r+1$ are computed as
\begin{equation*}
w^{r+1}_p = \bigcup_{i \in I} s^{r+1}_{p_i}
\cup \bigcap_{i \in I} \bar{c}^{r+1}_{p_i} - \{u^{r+1}_p\}
\end{equation*}
Fix process $p$.
We know that
\begin{equation*}
\bigcup_{i \in I} s^{r+1}_{p_i} = s^{r+1}_{p_\ell}
\end{equation*}
for some $\ell \in I$, since the simplexes are
ordered by inclusion.
Then $s^{r+1}_{p_\ell}$ is the largest simplex seen by $p$ in its snapshot.
Furthermore, we know that
$s^{r+1}_{p_\ell} \cup \bar{c}^{r+1}_{p_\ell}$ is a simplex,
by definition of the link of a simplex.
But
\begin{equation*}
\bigcap_{i \in I}
\bar{c}^{r+1}_{p_i} \subseteq \bar{c}^{r+1}_{p_\ell}
\end{equation*}
so
\begin{equation*}
w^{r+1}_p = \bigcup_{i \in I} s^{r+1}_{p_i} \cup \bigcap_{i \in I}
\bar{c}^{r+1}_{p_i} - \{u^{r+1}_p\} \subseteq s^r_{p_\ell} \cup
\bar{c}^{r+1}_{p_\ell}.
\end{equation*}
By downward closure, we conclude that $w^{r+1}_p$ is a simplex. 

By induction, all views and cores are simplexes.
\end{proof}

The next lemma states that the convergence complexes of the processes
are ordered by inclusion,
ensuring that the NCSA tasks solved by different processes are coherent,
even though they may have different convergence complexes. 

\begin{lemma}
\label{lemma:links}
The convergence complexes of participating processes for any given round are ordered by inclusion. 
\end{lemma}

\begin{proof}
Fix round $r$.
If exactly one or fewer processes have not decided,
then the claim is trivial.
So let $p_1$ and $p_2$ be distinct processes that have not decided.
Let $\cC_i$ denote the convergence complex for $p_i$.
We must show $\cC_1$ and $\cC_2$ are comparable.
Without loss of generality,
suppose $p_1$ takes a snapshot of the view arrays for rounds $1$ and $k$ before $p_2$.
Then $p_2$ sees more views than $p_1$,
so it computes a larger intersection for its core,
meaning that $c_2 \subseteq c_1$,
where $c_i$ is the core of $p_i$ in the current round.
Furthermore, if $Q_i$ is the set of processes that $p_i$ sees in round $1$,
then $Q_1 \subseteq Q_2$.
Since the link operator is order reversing in the first argument,
and order preserving in the second,
we have $\Lk(c_1, Q_1) \subseteq \Lk(c_2, Q_2)$.
We conclude that $\cC_1 \subseteq \cC_2$, which proves the lemma.
\end{proof}
Since convergence complexes are ordered by inclusion,
cores and participating sets are ordered in the same way.

The next lemma allows each process to pick a vertex of its own color
in its convergence complex when it starts LNCSA. 

\begin{lemma}
\label{lemma:vertex}
If process $p$ has not decided by round $r$,
then there is at least one vertex of its color in $\cC^r_p$.
\end{lemma}

\begin{proof}
It suffices to show that no vertex in $c^r_p$ has the same color as $p$.
The core $c^r_p$ is computed as the intersection of views that
$p$ sees in its snapshot of the view array in round $r$.
This includes $w^{r-1}_p$.
But $w^{r-1}_p$ cannot contain a vertex of color $p$,
since any such vertex is explicitly removed in the computation of
$w^{r-1}_p$.
Since $w^{r-1}_p$ does not contain a vertex of color $p$,
neither does $c^r_p$.
since $\cI$ is chromatic,
$\cC^r_p$ must contain at least one vertex of color $p$. 
\end{proof}

The next lemma states that a decision value reached in a given round
is contained in all cores and views of subsequent rounds.

\begin{lemma}
\label{lemma:stable}

If a process decides on a vertex,
then that vertex is contained in the views and cores of all processes
in all subsequent rounds.

\end{lemma}

\begin{proof}
Suppose process $p$ decides on vertex $u$ in round $r_p$.
We proceed by induction.
As the base case,
we show that every view computed in round $r_p$ contains $u$.
Let $p'$ be any other process that does not decide this round.
There are two cases:
when $p'$ computes its new view,
either $p'$ sees the simplex that $p$ wrote,
or $p'$ does not.
In the first case, $p'$ sees what $p$ wrote,
which is a simplex $\sigma$ that must contain $u$,
otherwise $p$ could not decide this round,
since $p$ must have seen $u$ in its own simplex.
So $p'$ includes $\sigma$ in the computation of its view,
so the view of $p'$ contains $u$.
Now consider the second case,
where $p'$ does not see the simplex that $p$ wrote.
But for $p$ to decide on $u$,
$u$ must have been contained in all previously written simplexes,
including the simplex of $p'$.
Since this simplex contains $u$, the view of $p'$ will also contain $u$.
In either case,
the view of $p'$ contains $u$. 

Now, inductively assume that all views from round $r \ge r_p$ contain vertex $u$.
We want to show that all views from round $r+1$ also contain $u$.
First, consider the cores computed in round $r+1$.
As intersections of views from round $r$,
all cores computed in round $r+1$ also contain $u$.
Furthermore, each process $p'$ computes its view in round $r+1$ as
\begin{equation*}
w^{r+1}_{p'} = \bigcup_{i \in I} s^{r+1}_{p_i} \cup \bigcap_{i \in I} \bar{c}^{r+1}_{p_i} - \{u^r_{p'}\}.
\end{equation*}
Since each core contains $u$,
we have
\begin{equation*}
u \in \bigcap_{i \in I} \bar{c}^{r+1}_{p_i} \subseteq w^{r+1}_{p'}
\end{equation*}
so  the view of $p'$ computed in round $r+1$ also contains $u$.
Therefore all views in round $r+1$ also contain $u$.
By induction, all views in all rounds after $r_p$ contain $u$.
Since cores are computed as intersections of views,
all cores in rounds $r > r_p$ also contain $u$. 
\end{proof}

We formally define LNCSA and give a protocol similar to that of NCSA.

\begin{lemma}
\label{lemma:LNCSA}
Fix round $r$, and suppose processes have computed their cores.
Then there is a wait-free immediate-snapshot protocol for converging
to a simplex on $\Lk(\bigcap_{k \in K} c^r_{p_k}, \bigcup_{k \in K} P^r_{p_k})$. 
\end{lemma}

\begin{proof}
Let $(\cI, \Div(\cI), \Div)$ be the chromatic simplex agreement task
we want to solve via the convergence algorithm.
Fix round $r$, and suppose the processes have just computed their links.
They will next collectively converge to a simplex that is consistent
with the smallest core among them.
That is, they will converge on the largest link among the processes.  
This is the LNCSA task.

We formally define the task. Each process's state is defined by its core $c$,
its participating set $p$,
and the starting vertex $v$ it chooses in its link
(which by Lemma~\ref{lemma:vertex} it can do),
so the input complex $\cA$ of LNCSA has vertex set consisting of triples $(v, c, P)$
such that $v \in P$ and $\Car(c, \cI) \subseteq P$.
A set of triples $\{(v_i, c_i, P_i)\}$ is a simplex in $\cA$ if the
$v_i$ all have distinct colors and $\{c_i\}$ and $\{P_i\}$ are ordered by
inclusion.
We require the $v_i$ to have different colors since each process
chooses a vertex of its own color,
and we require $\{c_i\}$ and $\{P_i\}$ to be ordered in the same way
by inclusion,
since they were computed using snapshots. 

So $\cA$ is a subcomplex of
\begin{equation*}
\Delta(V) \times \Bary(\Div(\cI)) \times \Bary(\cI),
\end{equation*}
where $V$ is the vertex set of $\cI$.
The output complex of LNCSA is $\cB = \Bary(\Div(\cI))$,
since each process chooses a simplex on $\Div(\cI)$.
Processes that execute in isolation stay where they are.
Otherwise, processes converge to a simplex in the largest link they see,
so if $\sigma = \{(v_i, c_i, P_i)\}$ is a simplex in $\cA$,
then we have $\Gamma(\sigma) = \{v\}$
if $\sigma = \{(v, c, P)\}$, and
$\Gamma(\sigma) = \Lk(\bigcap c_i, \bigcup P_i)$.
By Lemma \ref{lemma:links},
$\Gamma$ is monotonic, so it is a carrier map. 

Thus LNCSA is the task $(\cA, \cB, \Gamma)$. We show that this task is solvable by constructing a continuous
\begin{equation*}
f : |\cA| \rightarrow |\cB|
\end{equation*}
carried by $\Gamma$.
We induct on skeletons to construct an such an $f$ by defining carrier-preserving
\begin{equation*}
f^m : |\skel^m(\cA)| \rightarrow |\cB|
\end{equation*}
for each $m$.
As base case, we define $f^0$ as $f^0((v, c, P)) = v$ for each vertex
$(v, c, P)$ of $\cA$. 
The function $f^0$ is clearly continuous and carried by $\Gamma$.
Now inductively assume we have defined $f^m$.
We want to define
\begin{equation*}
f^{m+1} : |\skel^{m+1}(\cA)| \rightarrow |\cB|.  
\end{equation*}
Let $\{\sigma_i\}$ be the facets of $\skel^{m+1}(\cA)$,
and for each $i$, let $\sigma_i = \{(v_{ij}, c_{ij}, P_{ij})\}_{j \le m+2}$.
Let
\begin{equation*}
\overline{c}_i = \bigcap_j c_{ij} \text{\quad and \quad} \overline{P}_i = \bigcup_j P_{ij}.
\end{equation*}
By the inductive hypothesis,
$f^m$ is carried by $\Gamma$,
which implies that the image of $f^m$ is contained in
\begin{equation*}
\cL_i = \Lk(\overline{c}_i, \overline{P}_i).  
\end{equation*}
We also know that $\dim(\cL_i) \ge \dim(\sigma_i) = m$,
since each $v_{ij} \in \cL_i$, each $v_{ij}$ has different color,
and $\cL_i$ is chromatic.
Suppose $\dim(\overline{P}_i) = n$.
We know that $\cL_i$ a link in the subdivided simplex $\Bary(\Div(\overline{P}_i))$,
which is link-connected,
so since $\dim(\cL_i) \ge m$,
we conclude that $\cL_i$ is at least $n - 2 - (n - m - 1)$ = $(m - 1)$-connected.
Using this, along with the fact that $\im(f^m) \subseteq \cL_i$,
we can extend $f^m|_{\partial \sigma_i}$ to a function $f^{m+1}_i : |\sigma_i| \rightarrow |\cL_i|$.
Next, using the pasting lemma,
we can glue the $f^{m+1}_i$ together to obtain a map
\begin{equation*}
f^{m+1} : |\skel^{m+1}(\cA)| \rightarrow \cB  
\end{equation*}
extending $f^m$.
By construction, $f^{m+1}$ is carried by $\Gamma$.
By induction, we have obtained a continuous function $f : |\cA| \rightarrow |\cB|$.
By the simplicial approximation theorem,
the task $(\cA, \cB, \Gamma)$ is wait-free solvable using immediate snapshots.
\end{proof}

\subsection{Termination and Validity}
Next, we show that all processes eventually decide.

\begin{theorem}
\label{theorem:liveness}
All processes participating in the convergence algorithm eventually decide.
\end{theorem}

\begin{proof}
We show that at least one process decides each round.
Fix a round, and let $\{p_0, \ldots, p_\ell\}$ be the set of participating processes.
The processes each run LNCSA protocols over barycentric subdivisions
of the convergence complexes,
 and by Lemma \ref{lemma:LNCSA},
they converge to a simplex $\tau$ on the largest subcomplex,
say $\Bary(\cC)$.
By definition of $\Bary$, the simplex $\tau \in \Bary(\cA_{i_k})$
corresponds to a chain of simplexes $\sigma_{i_0} \subseteq \cdots \subseteq \sigma_{i_\ell}$.
Their intersection is $\sigma_{i_0}$, which is clearly nonempty,
so choose a vertex $\hat{v} \in \sigma_{i_0}$.
The color of $\hat{v}$, denoted as $\chi(\hat{v})$,
cannot be the color of any nonparticipating process,
because all subcomplexes $\Bary(\cC)$ are all subcomplexes of the
processes' carrier,
which contains only vertexes whose colors are those of the
participating processes,
since $\cI$ is chromatic.
Neither can $\chi(\hat{v})$ be the color of a process that
decided in a previous round.
Supposing not, this means the largest convergence complex contains
vertexes of color $\chi(\hat{v})$,
meaning that the corresponding core could not contain a vertex of
color $\chi(\hat{v})$,
since the link contains vertexes of color exactly those not in the
corresponding core.
This contradicts Lemma \ref{lemma:stable},
since decided vertexes must be contained in all cores. 

Let $\hat{p} \in P$ be the process whose color is $\chi(\hat{v})$,
and let $K \subseteq \{0, \ldots, \ell\}$ be the index set of
processes that $\hat{p}$ saw during its snapshot of the simplex array.
Then
\begin{equation*}
\bigcap_{k \le \ell} \sigma_{i_k} \subseteq \bigcap_{k \in K} \sigma_{i_k},
\end{equation*}
so $\hat{v} \in \bigcap_{k \in K} \sigma_{i_k}$.
Since $\hat{p}$ sees $\hat{v}$ in all simplexes in its snapshot,
$\hat{p}$ decides on $\hat{v}$.
It follows that at least one process decides in each round,
so all processes decide in at most $n+1$ rounds. 
\end{proof}

\begin{theorem}
\label{theorem:safety}

Participating processes converge to a simplex in their carrier.

\end{theorem}

\begin{proof}
We show that for each round $r$,
the set of vertexes that processes have decided form a simplex.
We argue by induction, starting at round $1$.
In round $1$,
it is clear that decision values form a simplex,
since processes decide on vertexes from the largest simplex chosen by running LNCSA.
All decision values in round $1$ must be contained in this largest simplex,
 hence they must form a simplex.

Fix round $r$, and inductively assume that the vertexes which processes
decided during the rounds up to $r$ form a simplex,
which we call $\tau_r$.
Consider round $r+1$.
Processes in this round choose vertexes on a simplex of the
barycentric subdivision of the largest convergence complex,
$\Bary(\cC)$, which is determined by the smallest core.
So vertexes on which processes decide this round are contained in an
ascending chain of simplexes in $\cC$.
Call the largest such simplex $\Sigma_{r+1}$.
Let $c_{r+1}$ be the smallest core,
corresponding to the largest convergence complex.
Then by definition of the link, $c_{r+1} \cup \Sigma_{r+1}$ is a simplex in $\cC$,
since the largest link is determined by simplex $c_r$.
From Lemma \ref{lemma:stable},
we must have $\tau_r \subseteq c_{r+1}$,
so $\tau_r \cup \Sigma_{r+1}$ is also a simplex in $\cC$,
by downward closure of simplicial complexes.
Let $\sigma_{r+1} \subseteq \Sigma_{r+1}$ be the set of all vertexes
that processes on decide during round $r+1$.
Again by downward closure,
\begin{equation*}
\tau_r \cup \sigma_{r+1} \subseteq \tau_r \cup \Sigma_{r+1}  ,
\end{equation*}
so $\tau_r \cup \sigma_{r+1}$ is a simplex.
But by definition of $\tau_r$,
$\tau_{r+1} = \tau_r \cup \sigma_{r+1}$,
where $\tau_{r+1}$ is exactly the set of vertexes on which processes have
decided up to round $r+1$.
So $\tau_{r+1}$ is also a simplex.
By induction,
it follows that the processes' decision values form a simplex. 

Processes can choose decision only values in their carrier,
since all links are computed relative to the observed participating set,
and all decision values are chosen from these links.
This completes the proof of the theorem,
and the proof of correctness of the convergence algorithm.
\end{proof}

\section{Application to more general tasks}

Recall that CSA over a chromatic subdivision $\Div(\K)$ is defined as
the task $(\K, \Div(\K), \Div)$.
The proof of the convergence algorithm shows that this task has a
wait-free read-write protocol.
To make the convergence algorithm work,
we required that $\Div(\K)$ be link-connected
so that processes can iteratively converge over ever smaller subcomplexes.
Phrased in a different way,
the convergence algorithm allows us to find a chromatic simplicial map
\begin{equation*}
\phi : \Ch^N(\cI) \rightarrow \Div \cI
\end{equation*}
carried by $\Div$,
Given the continuous (identity) map
$\mathrm{id} : |\cI| \rightarrow |\Div(\cI)|$ also carried by $\Div$.
In fact, link-connectivity is the essential property of the output complex.
In particular, the the convergence algorithm may be applied to more
general continuous functions $f : |\cI| \rightarrow |\cO|$ carried by some $\Gamma$.
Given the assumption that $\Gamma(\sigma)$ is link-connected for all
$\sigma \in \cI$,
we can use the convergence  algorithm to obtain a chromatic simplicial
map $\phi: \Ch^N(\cI) \rightarrow \cO$ also carried by $\Gamma$.
We state this observation as a theorem.  

\begin{theorem}
Let $f : |\cI| \rightarrow |\cO|$ be a continuous map between
chromatic complexes and let $\Gamma : \cI \rightarrow \cO$ be a
carrier map such that $\Gamma(\sigma)$ is link-connected for each
$\sigma \in \cI$.
Suppose $f$ is carried by $\Gamma$.
Then there exists a chromatic, carrier-preserving simplicial map
$\phi: \Ch^N(\cI) \rightarrow \cO$. 
\end{theorem}